\documentclass[12pt,a4paper,fleqn]{article}
\usepackage[utf8]{inputenc}
\usepackage{amsmath}
\usepackage{booktabs}
\usepackage{array}
\usepackage{url}
\usepackage{amsthm}
\usepackage{dcolumn}
\usepackage{endnotes}
\usepackage{units}
\usepackage{marvosym}
\usepackage[lines=10]{geometry}
\usepackage{longtable}
\usepackage{rotating}
\usepackage{expdlist}
\usepackage{floatrow}
\usepackage{pgfplots}
\usepackage{wrapfig}
\usepackage{amssymb}
\usepackage{gloss}
\usepackage{nomencl}
\usepackage{natbib}
\usepgfplotslibrary{groupplots}
\pgfmathdeclarefunction{gauss}{3}{
	\pgfmathparse{1/(#3*sqrt(2*pi))*exp(-((#1-#2)^2)/(2*#3^2))}%
}
\usetikzlibrary{decorations.pathreplacing,calligraphy,backgrounds}
\pgfplotsset{compat=1.16}
 
\pgfmathsetmacro\valueA{gauss(2.25, 2.25, 0.7)} 
\newtheorem{remark}{Remark}

\newtheorem{lemma}{Lemma}
\newtheorem{prop}{Proposition}
\newtheorem{coro}{Corollary}

    \makegloss
 
\renewcommand{\thesection}{\arabic{section}}

\usepackage{fancyhdr} 
\begin{document}
 
\begin{center}\huge{Information Aggregation in Markets \\with\\ Analysts, Experts, and Chatbots}\footnote{I thank seven anonymous referees for comments and suggestions which helped me to improve this paper substantially. First arXiv version 4.11.2024. There are no conflicts of interest. No funding was received. The paper does not use data.}
\end{center}
 
\begin{center} \textit{Wolfgang Kuhle}\\ \textit{Corvinus University of Budapest, Hungary, E-mail wkuhle@gmx.de\\ MEA, Max Planck Institute for Social Law and Social Policy, Munich, Germany}\end{center} 

\noindent\emph{\textbf{Abstract:} We study information aggregation when traders interpret public news, such as analyst reports, differently. This differential interpretation injects noise into equilibrium prices, compounding the informational bottleneck caused by noise traders. Next to analysts and experts, we consider the role of ``AI" chatbots, which (i) provide free financial advice and (ii) analyze traders' queries to learn about the asset's fundamental. Our model predicts that (i) chatbots aggregate information more efficiently than markets, (ii) increase trade volume, and (iii) enhance the informational content of prices. Ultimately, the AI platform extracts an absolute informational advantage over other market participants, highlighting a structural shift from information aggregation through market prices to ``AI" platforms.}
 
\section{Introduction}\label{Introduction}
 Most information aggregation models assume that public announcements are interpreted identically by all market participants. At the same time, \citet{Kan95}, \cite{Kandel99}, \cite{SHENG2015515} and others have empirically shown that trade volumes and stock returns indicate that public announcements are subject to differential interpretations by market participants. 
 
The present paper presents a simple model of information aggregation in which public news, such as an analyst report, is interpreted differently by different agents. In turn, we study how the differential interpretation of news impacts the informational content of prices. Moreover, we examine a marketplace where traders can send queries to an AI platform, which then provides ``advice." In equilibrium, the AI platform learns directly from agents' queries and thus centralizes and aggregates information more efficiently than markets with noise trader activity.  
 
We build our argument gradually. First, Section \ref{Section Model Perfect Information Aggregation} recalls the noise trader models of \citet{Gro76, Gro80, Hel80} to establish a benchmark of information aggregation through market prices. This baseline model abstracts from analyst reports and assumes that each trader receives a private signal over the asset's unobservable fundamental. 
 
Section \ref{Additional Information y Model} studies the role of analyst reports, which partially reveal the fundamental of the asset. In particular, we assume that each trader who reads the analyst report will interpret it slightly differently.\footnote{As mentioned earlier, \citet{Kan95} show that public announcements, such as earnings reports, are indeed interpreted differently by different traders. To see the same thing, one may also consider the heterogeneous interpretation of e.g. interest rate decisions by central banks: The announcement of a rate cut by the central bank may be seen as a positive signal by some traders. At the same time, other traders view it as a signal of a weaker economy, which is bad for asset prices. See also \citet{Sma23} and \citet{Han24}, who ask whether chatbots can improve our understanding of such central bank announcements.} In turn, we show that differential interpretation of information introduces a common noise component such that equilibrium price depends on the report's common error. The informational content of prices is thus constrained both by noise-trader supply and by interpretation noise.   
 
As an example of the differential interpretation of information, we may consider the report \cite{Dart25} on the Agnico Eagle Mines stock (AEM). This report mentions a number of facts related to the company, as well as Q1 earnings, which may be interpreted differently by different readers. For example, the report mentions that the stock trades at a price-earnings ratio of 20, which is elevated compared to its peers. Some investors may view this high price-earnings ratio as a signal that the stock is potentially overvalued, while others may view this high multiple as a signal that future earnings will grow faster than those of its competitors. Moreover, despite beating analyst expectations, Q1 output was slightly down. Again, investors may view the decline in output as a negative signal, while others may focus on the fact that output exceeded expectations.\footnote{Another example of differential interpretation of information is the discussion surrounding the \cite{Tesla26} investor report: Some analysts view the reported decline in car sales as an indication that Tesla may be a failing car business, while others think that the decline in car sales merely indicates that the company is refocusing on higher-margin businesses such as robots and its self-driving software.} 
 
Section \ref{Remarks and Extensions} introduces an ``AI'' chatbot that, instead of a human analyst, provides free financial advice. In equilibrium, traders submit queries that depend on their private signals, and the chatbot aggregates these queries to infer the asset's fundamental. Chatbots thus aggregate dispersed private information much as markets do. However, unlike prices, which aggregate information through demand and supply and are therefore constrained by noise traders, the chatbot observes traders’ queries directly, bypassing the noise-trader bottleneck. Section \ref{Section Information Centralization} shows that once the chatbot publishes the fundamental, prices no longer load on common interpretation noise, and traders are strictly better informed than in markets that rely on human analyst reports.
 
Section \ref{Section Volume} returns to the empirical motivation for differential interpretation and computes expected trading volume in all three economies. In particular, we show that differential interpretation of an analyst report raises volume above the \citet{Gro76,Gro80} benchmark, and that the AI economy generates the highest volume, even though it also features the most informative prices. Hence, high trading volume need not indicate that markets are poorly informed. Section \ref{Limitations and directions for future research} compares our approach in more detail with \citet{Gro76,Gro80} and with models of heterogeneous priors and public disagreement, discusses the role of our assumptions, namely a single dominant chatbot, rational Bayesian traders, normally distributed signals, and a static setting, and outlines directions for future research. Section \ref{Conclusion} concludes.
 
\paragraph{\textit{Related Literature:}} \citet{Kan95}, \citet{Kandel99}, and \citet{SHENG2015515} document that public announcements and news items are interpreted differently by different agents: \citet{Kan95} show that trading around earnings announcements is consistent with differential interpretation of public signals, \citet{Kandel99} find that forecasters interpret the same information differently when forecasting inflation, and \citet{SHENG2015515} construct a measure of differential interpretation from analysts' earnings forecasts that is positively related to trading volume. The present paper extends the standard model of information aggregation by \citet{Gro76,Gro80} and \citet{Hel80} to accommodate this empirical observation.\footnote{See \citet{Viv08} and \citet{Veronesi2009} for reviews of this literature and \citet{Atk01}, \citet{Ang06}, \citet{Kuh16Priors}, and \citet{Gor18} for applications of the baseline model of information aggregation in large markets.}
 
Our approach to differential interpretation differs from earlier theoretical work. The models of \citet{harrison1978speculative}, \citet{harris1993differences}, \citet{scheinkman2003overconfidence}, and \citet{GIZATULINA2019161} attribute disagreement to heterogeneous priors or to publicly known differences of opinion regarding the model's parameters. In \citet{harris1993differences}, for example, it is common knowledge that there are two groups of traders who hold conflicting views regarding the precision of a public signal. In our model, by contrast, each trader \emph{privately} forms his own view of what a public report means, so that each trader's reading of the report is private information. Section \ref{Limitations and directions for future research} discusses these differences in more detail.
 
These models, together with the empirical work of \citet{Kan95} and \citet{SHENG2015515}, are primarily concerned with how differences of opinion generate trade and, in \citet{harrison1978speculative} and \citet{scheinkman2003overconfidence}, speculative price premia. The present paper focuses instead on how differential interpretation affects the informational content of prices: the report's common error enters equilibrium prices, such that prices contain both noise-trader supply and common interpretation noise. At the same time, our model remains consistent with the evidence on trade, since differential interpretation raises trading volume above the \citet{Gro76,Gro80} benchmark.
 
Finally, the present paper connects the classic literature on information aggregation with recent research on AI-based tools in financial markets. \citet{Sma23} and \citet{Han24} ask whether large language models can interpret central bank communication, which is precisely the type of public news that is subject to differential interpretation; \citet{Mil23} study how users take advice from chatbots; and \citet{Sun24}, \citet{Kou25}, and \citet{Fen25} review applications of alternative data, financial technology, and ChatGPT in finance. Our model emphasizes a different aspect of AI tools: chatbots serve a role similar to markets. Rather than submitting demand functions to the marketplace, traders submit queries to a chatbot. Because these queries depend on traders' private information, the chatbot can centralize and aggregate traders' dispersed private knowledge and infer the asset's fundamental without the noise-trader bottleneck that limits prices. Taking this view, the present model predicts that chatbots may rival marketplaces, such as stock exchanges, as providers of financial information.

\section{Information Aggregation}\label{Section Model Perfect Information Aggregation}
 
Following \citet{Gro80} and \citet{Hel80}, we assume that there exists a measure-one continuum of agents indexed by $i\in[0,1]$. There exists a stochastic aggregate supply of shares $K^S$. Each share offers an unobservable return $\theta$. Otherwise, all model coefficients are common knowledge among agents. Agents hold a uniform, uninformative prior over $\theta$. Moreover, each agent receives a private signal regarding $\theta$:
\begin{eqnarray} 
    x_i=\theta+\sigma_x\xi_i, \quad \xi_i\sim\mathcal{N}(0,1) \label{signalx} 
\end{eqnarray}
The signal noise $\xi_i$ is distributed independently among agents and independently of $\theta$ and the other aggregate noise terms that are introduced below. Agents maximize CARA utility. We denote the coefficient of absolute risk aversion by $\gamma$, and the market price by $P$. Using their information, agents choose asset demand $k_i$ to maximize expected utility:
\begin{eqnarray} 
    U_i=-E\left[\left.e^{-\gamma(k_i(\theta-P))}\right|x_i,P\right]\label{Utility} 
\end{eqnarray}
Finally, aggregate supply is given by:
\begin{eqnarray}
    K^S=\sigma_{\eta}\eta, \quad \eta\sim\mathcal{N}(0,1) \label{KS}
\end{eqnarray}
where $\sigma_{\eta}\eta$ represents unobservable noise trader activity as in \citet{Gro80} and \citet{Hel80}. To solve for the market's equilibrium, we guess and verify that prices are given by a linear function:
\begin{eqnarray} 
    P=a\theta+b\eta+c \label{guess1} 
\end{eqnarray}
Taking into account (\ref{guess1}), the market price, which is endogenous public information, is informationally equivalent to a signal $\tilde{Z}$:
\begin{eqnarray} 
    \tilde{Z}:=\frac{P-c}{a}=\theta+\frac{b}{a}\eta  \label{Zsignalfullyrevealing}
\end{eqnarray}
Utility (\ref{Utility}) implies that traders' demand functions are:
\begin{eqnarray} 
    k_i=\frac{E[\theta|x_i,\tilde{Z}]-P}{\gamma \mathrm{Var}(\theta|x_i,\tilde{Z})}=\frac{\frac{\alpha_x}{\alpha}x_i+\frac{\alpha_{\tilde{Z}}}{\alpha}\tilde{Z}-P}{\gamma\frac{1}{\alpha}}, \quad \alpha_x:=\frac{1}{\sigma^2_x}, \quad \alpha_{\tilde{Z}}:=\frac{a^2}{b^2}, \quad \alpha=\alpha_x+\alpha_{\tilde{Z}} \label{demand1} 
\end{eqnarray}
Finally, market clearing requires $K^D=K^S$, such that:
\begin{eqnarray} 
    K^D=\int_0^1k_idi=\frac{\alpha_x\theta+\alpha_{\tilde{Z}}\left(\theta+\frac{b}{a}\eta\right)-\alpha P}{\gamma}=\sigma_{\eta}\eta \label{Equi}
\end{eqnarray} 
Rewriting (\ref{Equi}) yields an equilibrium price:
\begin{eqnarray}
    P=\theta+\frac{1}{\alpha}\left(\frac{a}{b}-\gamma \sigma_{\eta}\right)\eta \label{equilibrium} 
\end{eqnarray} 
Comparing (\ref{equilibrium}) with our initial guess (\ref{guess1}) yields:
 
\begin{prop}[\citealt{Hel80}]\label{Hellwig1980}
Equilibrium prices (\ref{equilibrium}) are informationally equivalent to a signal $\tilde{Z} = \theta - \gamma \sigma_x^2 \sigma_{\eta} \eta$, which is only partially revealing of the asset's fundamental $\theta$ due to the aggregate supply noise $\eta$. Agents' posterior variance is $\mathrm{Var}(\theta|x_i,P)=\frac{1}{\alpha}=\frac{1}{\frac{1}{\sigma_x^2}+\frac{1}{(\gamma\sigma_x^2\sigma_{\eta})^2}}$.
\end{prop} 
 
\begin{proof}
The precision $\alpha_{\tilde{Z}}=a^2/b^2$ in (\ref{demand1}) is finite only if $b\neq0$. We therefore consider equilibria with $b\neq0$; the case $b=0$, in which prices are fully revealing, is not an equilibrium, for the reasons given in Remark \ref{Remark1} below. Matching coefficients between the price guess (\ref{guess1}) and the derived market-clearing price (\ref{equilibrium}) requires $a=1$, $c=0$, and $b = \frac{1}{\alpha}\left(\frac{1}{b}-\gamma\sigma_{\eta}\right)$. Substitute $\alpha = \alpha_x + \frac{1}{b^2}$ into the condition for $b$:
\begin{eqnarray*}
    b\left(\alpha_x + \frac{1}{b^2}\right) &=& \frac{1}{b} - \gamma \sigma_{\eta} \\
    \alpha_x b + \frac{1}{b} &=& \frac{1}{b} - \gamma \sigma_{\eta} \\
    \alpha_x b &=& -\gamma \sigma_{\eta}
\end{eqnarray*}
Solving for $b$ yields $b = -\frac{\gamma\sigma_{\eta}}{\alpha_x} = -\gamma\sigma_x^2\sigma_{\eta}$, which is indeed nonzero since $\sigma_\eta>0$. Hence, the equilibrium price depends on the unobservable noise $\eta$, preventing perfect inference of $\theta$. We thus have
\begin{eqnarray}
    \mathrm{Var}(\theta|x_i,P)=\frac{1}{\alpha}=\frac{1}{\frac{1}{\sigma_x^2}+\frac{1}{b^2}}=\frac{1}{\frac{1}{\sigma_x^2}+\frac{1}{(\gamma\sigma_x^2\sigma_{\eta})^2}} \label{PosteriorVarianceBenchmark}
\end{eqnarray}
\end{proof}

\section{Differential Interpretation of Information}\label{Additional Information y Model}
Suppose now that agents, in addition to their signal $x_i$, receive another private signal $y_i$:\begin{eqnarray} y_i=r+\sigma_y\tau_i=\theta+\sigma_{\varepsilon}\varepsilon+\sigma_y\tau_i, \quad \tau_i\sim\mathcal{N}(0,1), \quad \varepsilon\sim\mathcal{N}(0,1) \label{Published Signal}  \end{eqnarray}We interpret signal $y_i$ as the noisy information that agents receive when reading a news report $r=\theta+\sigma_{\varepsilon}\varepsilon$. That is, reading is noisy, and each agent extracts a private reading $y_i$ rather than observing the report's true meaning $r$ directly. Because the report's true meaning $r$ is never observed directly, traders form their expectations using only $x_i$, $y_i$, and the endogenous market price $P$.

\noindent Given this information, agents maximize expected utility:
\begin{eqnarray} U_i=-E[e^{-\gamma(k_i(\theta-P))}|x_i,y_i,P]=-e^{-\gamma k_i(E[\theta|x_i,y_i,P]-P)+k_i^2\frac{\gamma^2}{2}\mathrm{Var}(\theta|x_i,y_i,P)} \label{utility1}\end{eqnarray}
such that individual demand is:
\begin{eqnarray} k_i=\frac{E[\theta|x_i,y_i,P]-P}{\gamma \mathrm{Var}(\theta|x_i,y_i,P)}\label{demand} \end{eqnarray}
To compute agents' posterior expectations, we guess and verify a price function:
\begin{eqnarray}  P=a\theta+b\eta+g\varepsilon+c \label{guess2}\end{eqnarray}
It follows from (\ref{guess2}) that prices $P$ carry the same information regarding $\theta$ as a signal:
\begin{eqnarray} Z:=\frac{P-c}{a}=\theta+\frac{b}{a}\eta+\frac{g}{a}\varepsilon \label{Zsignal} \end{eqnarray} 
Taken together, traders' private signals $x_i,y_i$, and the price signal (\ref{Zsignal}), yield posteriors:\footnote{See, e.g., \cite{Rai61}, p. 250 for the standard results on posteriors of normal distributions.}
\begin{equation}
E[\theta|x_i,y_i,Z]
=x_i+\frac{\sigma_x^2}{\Lambda}
\left[
\left(\frac{b^2}{a^2}+\frac{g^2}{a^2}-\sigma_{\varepsilon}\frac{g}{a}\right)(y_i-x_i)
+\left(\sigma_{\varepsilon}^2+\sigma_y^2-\sigma_{\varepsilon}\frac{g}{a}\right)(Z-x_i)
\right]
\label{Posterior mean}
\end{equation}
 
\begin{equation}
\mathrm{Var}(\theta|x_i,y_i,Z)
=\frac{\sigma_x^2}{\Lambda}
\left[
\left(\left(\frac{b}{a}\right)^2+\frac{g^2}{a^2}\right)
(\sigma_{\varepsilon}^2+\sigma_y^2)
-\sigma_{\varepsilon}^2\frac{g^2}{a^2}
\right]
\label{Posterior Variance}
\end{equation}
 
\begin{equation}
\Lambda
=(\sigma_{\varepsilon}^2+\sigma_y^2)
\left(\left(\frac{b}{a}\right)^2+\frac{g^2}{a^2}\right)
-\sigma_{\varepsilon}^2\frac{g^2}{a^2}
+\sigma_x^2
\left(
\left(\frac{b}{a}\right)^2+\frac{g^2}{a^2}
+\sigma_{\varepsilon}^2+\sigma_y^2
-2\sigma_{\varepsilon}\frac{g}{a}
\right),
\quad \Lambda>0,
\label{determinant}
\end{equation}
Appendix A derives \eqref{Posterior mean}--\eqref{determinant} and the aggregate demand \eqref{AggregateDemand} below in detail.
 
\noindent Substituting the posteriors (\ref{Posterior mean})--(\ref{determinant}) into individual demand (\ref{demand}), and integrating over the continuum of agents, yields aggregate demand:
\begin{align}
K^D &= \int_0^1 k_idi
=\int_0^1\frac{E[\theta|x_i,y_i,Z]-P}{\gamma \mathrm{Var}(\theta|x_i,y_i,Z)}di \nonumber\\
&=\frac{
\theta
+\frac{\sigma_x^2}{\Lambda}
\left[
\left(\sigma_{\varepsilon}^2+\sigma_y^2-\sigma_{\varepsilon}\frac{g}{a}\right)
\left(\frac{b}{a}\eta+\frac{g}{a}\varepsilon\right)
+
\left(
\left(\frac{b}{a}\right)^2+\left(\frac{g}{a}\right)^2
-\sigma_{\varepsilon}\frac{g}{a}
\right)\sigma_{\varepsilon}\varepsilon\right]-P}{\gamma \mathrm{Var}(\theta|x_i,y_i,Z)} \label{AggregateDemand}
\end{align}
 
Market clearing requires $K^D=\sigma_{\eta}\eta$, such that:
\begin{eqnarray}
P=&&\theta
+\frac{\sigma_x^2}{\Lambda}
\left[
\left(\sigma_{\varepsilon}^2+\sigma_y^2-\sigma_{\varepsilon}\frac{g}{a}\right)
\left(\frac{b}{a}\eta+\frac{g}{a}\varepsilon\right)+
\left(
\left(\frac{b}{a}\right)^2
+\left(\frac{g}{a}\right)^2
-\sigma_{\varepsilon}\frac{g}{a}
\right)
\sigma_{\varepsilon}\varepsilon
\right]
\nonumber\\&&-\gamma\sigma_{\eta}\eta
\mathrm{Var}(\theta|x_i,y_i,Z).
\label{Price2}
\end{eqnarray}
 
Comparing the price function (\ref{Price2}) with our initial guess (\ref{guess2}), taking into account (\ref{Posterior mean})-(\ref{determinant}), yields the model coefficients $a=1,b,g$ and $c=0$: 
 
\begin{prop}\label{Two equilibria Proposition}
Every linear equilibrium has $a=1$, $c=0$, and a value $g\neq 0$ that satisfies the cubic equation
\begin{eqnarray}
A_3g^3+A_2g^2+A_1g+A_0=0, \label{CubicG}
\end{eqnarray}
where
\begin{eqnarray*}
A_3&=&\sigma_x^4\gamma^2\sigma_{\eta}^2\left(\sigma_x^2+\sigma_y^2+\sigma_{\varepsilon}^2\right)+\sigma_{\varepsilon}^2\left(\sigma_x^2+\sigma_y^2\right),\\
A_2&=&-\sigma_{\varepsilon}\sigma_x^2\left[\gamma^2\sigma_{\eta}^2\sigma_x^2\left(3\sigma_x^2+2\sigma_y^2+2\sigma_{\varepsilon}^2\right)+2\sigma_{\varepsilon}^2\right],\\
A_1&=&\gamma^2\sigma_{\eta}^2\sigma_{\varepsilon}^2\sigma_x^4\left(3\sigma_x^2+\sigma_y^2+\sigma_{\varepsilon}^2\right)+\sigma_{\varepsilon}^4\sigma_x^2,\\
A_0&=&-\gamma^2\sigma_{\eta}^2\sigma_{\varepsilon}^3\sigma_x^6,
\end{eqnarray*}
and $b$ is recovered from $g$ via
\begin{eqnarray}
b^2=\frac{-g\left[\left(\sigma_x^2+\sigma_y^2\right)g^2-2\sigma_{\varepsilon}\sigma_x^2g+\sigma_{\varepsilon}^2\sigma_x^2\right]}{\left(\sigma_{\varepsilon}^2+\sigma_x^2+\sigma_y^2\right)g-\sigma_{\varepsilon}\sigma_x^2}. \label{BfromG}
\end{eqnarray}
\end{prop}
\begin{proof}
Matching the $\theta$-coefficient of (\ref{Price2}) to (\ref{guess2}) yields $a=1$; matching constants, given that (\ref{KS}) has zero mean, yields $c=0$. Matching the $\eta$- and $\varepsilon$-coefficients yields, with $V_y:=\sigma_{\varepsilon}^2+\sigma_y^2$, $V_z:=b^2+g^2$, $\kappa:=\sigma_{\varepsilon}g$, and $Q:=\frac{\sigma_x^2}{\Lambda}(V_y-\kappa)$, the posterior weight on $Z$ in (\ref{Posterior mean}),
\begin{eqnarray*}
b(1-Q)=-\gamma\sigma_{\eta}\mathrm{Var}(\theta|x_i,y_i,Z), \qquad g(1-Q)=\sigma_x^2\sigma_{\varepsilon}(V_z-\kappa)/\Lambda.
\end{eqnarray*}
Clearing denominators, the second equation is linear in $b^2$ (given $g$) and yields (\ref{BfromG}); if $g=0$, (\ref{BfromG}) gives $b^2=0$. By Remark \ref{Remark1}, the resulting trivial solution $b=g=0$ is not an equilibrium, so every equilibrium has $g\neq0$. Squaring the first equation, substituting (\ref{BfromG}), and clearing denominators eliminates $b$ and leaves a polynomial in $g$ alone that is proportional to
\begin{eqnarray*}
g^2\left(\sigma_{\varepsilon}^2+\sigma_y^2-\sigma_{\varepsilon}g\right)^2\left(A_3g^3+A_2g^2+A_1g+A_0\right).
\end{eqnarray*}
 The factor $\left(\sigma_{\varepsilon}^2+\sigma_y^2-\sigma_{\varepsilon}g\right)^2$ vanishes only at $g=(\sigma_{\varepsilon}^2+\sigma_y^2)/\sigma_{\varepsilon}>\sigma_{\varepsilon}$, which lies outside the range of equilibrium values of $g$ established in Lemma \ref{GBoundLemma} (below), and the factor $g^2$ corresponds to $g=0$. Dividing out both factors leaves the cubic (\ref{CubicG}). Since $A_0\neq 0$, $g=0$ is not a root of (\ref{CubicG}); the trivial solution $g=0$ corresponds to the factor $g^2$ that is divided out. Because the first equation was squared, (\ref{CubicG}) and (\ref{BfromG}) determine $b$ only up to its sign, which will be pinned down in Proposition \ref{Prop_Uniqueness_Noisy}.
\end{proof}
 
\noindent Regarding Proposition \ref{Two equilibria Proposition} we remark:
\begin{remark}\label{Remark1}
    The coefficient equations also have the trivial solution $a=1$, $b=g=c=0$, for which $P=\theta$. This solution is not an equilibrium. Since the price reveals $\theta$, $E[\theta|x_i,y_i,P]-P=0$ and $\mathrm{Var}(\theta|x_i,y_i,P)=0$, such that individual demand \eqref{demand} is of the form $0/0$: every trader is indifferent among all positions, and no demands based on $(x_i,y_i,P)$ can absorb the noise supply $\sigma_\eta\eta$, which is independent of $(\theta,\varepsilon)$. This is the familiar difficulty with fully revealing prices: if prices reveal $\theta$, demands no longer depend on private information, such that it is unclear how this information enters prices in the first place.
\end{remark}
 
\begin{lemma}\label{GBoundLemma}
In any equilibrium, the coefficients $b$ and $g$ satisfy $b^2>0$ and $g\in(0,\bar g)$, where
\begin{eqnarray}
\bar g:=\frac{\sigma_{\varepsilon}\sigma_x^2}{\sigma_{\varepsilon}^2+\sigma_x^2+\sigma_y^2}<\sigma_{\varepsilon}.
\label{GBar}
\end{eqnarray}
In particular, $g\in(0,\sigma_{\varepsilon})$.
\end{lemma}
\begin{proof}
Write (\ref{BfromG}) as $b^2=-gN(g)/D(g)$, where
\begin{eqnarray*}
N(g):=\left(\sigma_x^2+\sigma_y^2\right)g^2-2\sigma_{\varepsilon}\sigma_x^2 g+\sigma_{\varepsilon}^2\sigma_x^2, \qquad D(g):=\left(\sigma_{\varepsilon}^2+\sigma_x^2+\sigma_y^2\right)g-\sigma_{\varepsilon}\sigma_x^2.
\end{eqnarray*}
As a quadratic in $g$, $N$ has positive leading coefficient $\sigma_x^2+\sigma_y^2>0$ and discriminant
\begin{eqnarray*}
\operatorname{disc}(N)=4\sigma_{\varepsilon}^2\sigma_x^4-4\left(\sigma_x^2+\sigma_y^2\right)\sigma_{\varepsilon}^2\sigma_x^2=-4\sigma_{\varepsilon}^2\sigma_x^2\sigma_y^2<0,
\end{eqnarray*}
since $\sigma_y>0$. Hence $N(g)>0$ for every $g\in\mathbb{R}$, and the sign of $b^2$ coincides with the sign of $-g/D(g)$.
 
Suppose $g\leq0$. Then $D(g)=\left(\sigma_{\varepsilon}^2+\sigma_x^2+\sigma_y^2\right)g-\sigma_{\varepsilon}\sigma_x^2<0$, since the first term is non-positive and the second strictly negative. Thus $-g/D(g)\leq0$, so $b^2\leq0$; as $b^2\geq0$ by definition, this forces $b^2=0$ and, since $N(g)>0$, $g=0$. This is the trivial solution, which by Remark \ref{Remark1} is not an equilibrium.
 
Hence any equilibrium has $g>0$. At $D(g)=0$, i.e., at $g=\bar g$, the second equilibrium condition, which is linear in $b^2$, has no solution, since $gN(g)\neq0$. For $g>0$, $-g<0$, so $b^2\geq0$ requires $D(g)<0$, i.e.,
\begin{eqnarray*}
g<\frac{\sigma_{\varepsilon}\sigma_x^2}{\sigma_{\varepsilon}^2+\sigma_x^2+\sigma_y^2}=\bar g.
\end{eqnarray*}
Finally,
\begin{eqnarray*}
\sigma_{\varepsilon}-\bar g=\frac{\sigma_{\varepsilon}\left(\sigma_{\varepsilon}^2+\sigma_y^2\right)}{\sigma_{\varepsilon}^2+\sigma_x^2+\sigma_y^2}>0,
\end{eqnarray*}
so $\bar g<\sigma_{\varepsilon}$, and $g\in(0,\bar g)\subset(0,\sigma_{\varepsilon})$. On this interval, $N(g)>0$ and $D(g)<0$, such that $b^2=-gN(g)/D(g)>0$.
\end{proof}
 
\begin{remark}
Lemma \ref{GBoundLemma} gives an economic bound on the price's dependence on the common report error $\varepsilon$. Full pass-through of the raw report $r=\theta+\sigma_{\varepsilon}\varepsilon$ into the price would correspond to $g=\sigma_{\varepsilon}$. Instead, $\bar g=\sigma_{\varepsilon}\bar w_y$, where $\bar w_y:=\sigma_x^2/(\sigma_{\varepsilon}^2+\sigma_x^2+\sigma_y^2)\in(0,1)$ is the weight that a trader who combines only his two private signals $x_i$ and $y_i$ places on $y_i$. Hence $\bar g$ is the loading on $\varepsilon$ of a price that merely averages traders' private-signal estimates, $\theta+\bar w_y\sigma_{\varepsilon}\varepsilon$, which is the limit in which prices carry no information ($b^2\to\infty$ as $g\uparrow\bar g$). Because an informative price also induces traders to shift weight from $y_i$ towards $P$, equilibrium $g$ lies strictly below $\bar g$.
\end{remark}
 
\begin{prop}\label{Prop_Uniqueness_Noisy}
For every $\sigma_\varepsilon>0$, the cubic \eqref{CubicG} has exactly one root in $(0,\bar g)$. Hence, the economy with common noise has a unique linear equilibrium $(b,g)$.
\end{prop}
\begin{proof}
By Lemma \ref{GBoundLemma}, every equilibrium has $g\in(0,\bar g)$. Let $T:=\sigma_\varepsilon^2+\sigma_x^2+\sigma_y^2$, $V_y=\sigma_\varepsilon^2+\sigma_y^2$ as in the proof of Proposition \ref{Two equilibria Proposition}, and
\[
\psi:=\frac{\gamma^2\sigma_\eta^2\sigma_x^2 V_y}{\sigma_\varepsilon^2}>0.
\]
The map $u\mapsto g=\bar g\,u/(1+u)$ is a strictly increasing bijection from $(0,\infty)$ onto $(0,\bar g)$. Substituting it into \eqref{CubicG} and multiplying by $(1+u)^3T^3>0$ yields
\[
c_3u^3+c_2u^2+c_1u+c_0=0,
\]
where
\begin{eqnarray*}
c_3&=&\sigma_\varepsilon^5\sigma_x^4\left(V_y^2+\sigma_x^2\sigma_y^2\right)>0,\\
c_2&=&-\sigma_\varepsilon^5\sigma_x^4\,V_y T\,(\psi-2),\\
c_1&=&-\sigma_\varepsilon^5\sigma_x^4\,T^2\,(2\psi-1),\\
c_0&=&-\gamma^2\sigma_\eta^2\sigma_\varepsilon^3\sigma_x^6T^3<0.
\end{eqnarray*}
Hence $\mathrm{sign}(c_2)=\mathrm{sign}(2-\psi)$ and $\mathrm{sign}(c_1)=\mathrm{sign}(1-2\psi)$. If $c_2<0$, then $\psi>2$, so $c_1<0$. The sign pattern $(+,-,+,-)$ is therefore impossible, and, ignoring zero coefficients, the sequence $(c_3,c_2,c_1,c_0)$ exhibits exactly one sign change for every $\psi>0$. By Descartes' rule of signs, the cubic in $u$ has exactly one positive root, and thus \eqref{CubicG} has exactly one root $g^*\in(0,\bar g)$. At $g^*$, \eqref{BfromG} yields $b^2>0$, since, by Lemma \ref{GBoundLemma}, $N(g^*)>0$ and $D(g^*)<0$. The sign of $b$ is pinned down by $b(1-Q)=-\gamma\sigma_\eta \mathrm{Var}(\theta|x_i,y_i,Z)<0$ together with $1-Q>0$ (see \eqref{OneMinusQ_app} in Appendix B), so $b<0$. The equilibrium is therefore unique.
\end{proof}
 
\begin{remark}
Although \eqref{CubicG} may have up to three positive real roots, exactly one of them lies in $(0,\bar g)$; any root above $\bar g$ would imply $b^2<0$ by \eqref{BfromG} and therefore does not constitute an equilibrium.
\end{remark}
 
\begin{remark}\label{Remark_IdenticalInterpretation}
Differential interpretation ($\sigma_y>0$) is essential for the report's common error to confound prices. If all traders read the report identically ($\sigma_y=0$), the report $r$ is exogenous public information that every trader observes. Appendix C shows that the price then still depends on $\varepsilon$, but traders filter the report out of the price, and the price net of the report has exactly the precision $1/(\gamma\sigma_x^2\sigma_\eta)^2$ of the benchmark price in Proposition \ref{Hellwig1980}. Identically interpreted public information thus does not change the precision of the price as a signal of $\theta$. With $\sigma_y>0$, by contrast, traders do not observe $r$, and the report's common error becomes an unobservable component of the price.
\end{remark}
 
\section{``AI'' Chatbots as Experts}\label{Remarks and Extensions}
 
Next to analysts and experts, traders may turn to ``AI'' chatbots for financial advice. Like analyst reports, chatbot answers are public information that is subject to differential interpretation. To study the role of chatbots, we assume that there exists a single chatbot.\footnote{That is, we study a ``chatbot monopoly,'' in which one AI has successfully captured a dominant position in the market for financial advice, centralizing the flow of user queries.} Moreover, we assume that:
\begin{enumerate}
    \item The chatbot processes data efficiently, i.e., it uses Bayes' formula to compute posteriors.
    \item Prior to the realization of the market price, agents $i\in[0,1]$ use their private information $x_i=\theta+\sigma_x\xi_i$, with $\xi_i\sim\mathcal{N}(0,1)$, to formulate queries $q_i=f(x_i)$, which they then submit to the AI chatbot. 
    \item Queries $q_i=f(x_i)$ can be inverted such that $x_i=f^{-1}(q_i)$.\footnote{Alternatively, our results also obtain if the chatbot can use queries to extract a noisy signal $\tilde{q}_i=x_i+\sigma_q\zeta_i$, $\zeta_i\sim\mathcal{N}(0,1)$, regarding the agents' private information $x_i$. If the $\zeta_i$ are independent across agents and of all other shocks, then $\int_0^1\tilde{q}_i\,di=\theta$, by the same argument as in the proof of Proposition \ref{Prop_AI_Learning}, so that these signals suffice to learn $\theta$.}
    \item The chatbot rewards traders' queries by publishing its posterior expectation of the fundamental. Traders interpret this public answer with idiosyncratic noise, extracting a private reading $y_i = \mathit{AI\_Answer} + \sigma_y\tau_i$, with $\tau_i\sim\mathcal{N}(0,1)$.
\end{enumerate}
 
\paragraph{The AI learns the fundamental efficiently from traders' queries.} Because the chatbot observes each trader's query directly, it bypasses the noisy channel of aggregate supply and demand. We thus have:

\begin{prop}\label{Prop_AI_Learning}
The AI chatbot perfectly infers the fundamental $\theta$ from trader queries, centralizing information strictly more efficiently than the rational expectations market price.
\end{prop}

\begin{proof}
Suppose each trader $i \in [0,1]$ submits a query $q_i = f(x_i)$ based on their private signal $x_i = \theta + \sigma_x \xi_i$, where $\xi_i \sim \mathcal{N}(0,1)$. Because the AI knows the invertible mapping $f(\cdot)$, it recovers the continuum of private signals $x_i$. Integrating these signals over the population yields:
\begin{equation}
    \int_0^1 x_i \, di = \int_0^1 (\theta + \sigma_x \xi_i) \, di = \theta + \sigma_x \int_0^1 \xi_i \, di
\end{equation}
By the exact law of large numbers, $\int_0^1 \xi_i \, di = 0$ almost surely, allowing the AI to learn $\theta$ with zero variance. In contrast, the market price is informationally equivalent to the signal $\tilde{Z} = \theta + \frac{b}{a}\eta$ in (\ref{Zsignalfullyrevealing}), which is persistently confounded by the aggregate supply shock $\eta$. Likewise, the price in the economy with an analyst report is informationally equivalent to $Z=\theta+\frac{b}{a}\eta+\frac{g}{a}\varepsilon$ in (\ref{Zsignal}), leaving residual uncertainty about $\theta$ in both market economies.
\end{proof}

Because the AI perfectly infers the fundamental, its published $\mathit{AI\_Answer} = \theta$. Consequently, traders extract the private reading $y_i = \theta + \sigma_y\tau_i$. Unlike the human analyst report from Section \ref{Additional Information y Model}, the chatbot's answer contains no common error ($\sigma_\varepsilon = 0$), which fundamentally alters the informational efficiency of the market.
\paragraph{Robustness: coarse, non-invertible queries.} In the baseline case above, each signal $x_i$ induces an invertible query $f(x_i)$. Instead, let us now assume that language bandwidth is limited, so that traders with \textit{similar} signals $x_i$ submit \textit{identical} questions. In this case, there no longer exists a globally invertible map $f(\cdot)$ between signals and queries. 
 
To show that the chatbot can still infer $\theta$ under such coarse communication, consider the extreme case where agents face a binary choice of questions. Specifically, for all signals $x_i < x^*$, agents submit question $A$, and for all signals $x_i > x^*$, they submit question $B$. The chatbot observes the aggregate mass of agents, $\lambda$ and $1-\lambda$, who submitted questions $A$ and $B$, respectively. As illustrated in Figure \ref{fig:CoarseQueries}, the chatbot can use this relative frequency to infer the fundamental: 
$\lambda = \Pr(x_i < x^* \mid \theta) = \Phi\left(\frac{x^* - \theta}{\sigma_x}\right)$, 
where $\Phi(\cdot)$ is the cumulative distribution function of the standard normal distribution. Inverting this relationship yields the chatbot's estimate of the fundamental:
\begin{equation}
    \theta = x^* - \sigma_x \Phi^{-1}(\lambda).
\end{equation}
 
\begin{figure}[htbp]
\centering
\begin{tikzpicture}
  \begin{axis}[
    domain=-2:6,
    samples=200,
    axis lines=middle,
    xlabel={$x$},
    ylabel={$\phi(x|\theta)$},
    height=6cm,
    width=10cm,
    xtick={0},
    ytick=\empty,
    enlargelimits=0.1,
    clip=false
  ]
    \addplot[thick,blue] {1/sqrt(2*pi) * exp(-((x-2)^2)/2)};
    
    \addplot [
      draw=none,
      domain=-2:1,
      samples=200,
      fill=black,
      opacity=0.3
    ]
    {1/sqrt(2*pi) * exp(-((x-2)^2)/2)} \closedcycle;
 
    \draw[dashed,red] (axis cs:2,0) -- (axis cs:2,0.4);
    \node[below,red] at (axis cs:2,0) {$\theta$};
 
    \draw[dashed,black] (axis cs:1,0) -- (axis cs:1,0.245);
    \node[below,black] at (axis cs:1,0) {$x^*$};
  \end{axis}
\end{tikzpicture}
\caption{The shaded area under the normal density function $\phi(x|\theta)$ represents the share $\lambda$ of type $A$ questions, which are submitted when the asset's fundamental is $\theta$. Thus, a chatbot observing this share $\lambda$ can infer the fundamental $\theta$.}
\label{fig:CoarseQueries}
\end{figure}
 
\section{Information Centralization via AI}\label{Section Information Centralization}
 
We now contrast the informational efficiency of the market mechanism with that of an AI chatbot. In standard large markets, prices act as the central aggregator of dispersed private information. However, this aggregation is inherently bottlenecked by unobservable noise trader demand, which prevents prices from fully revealing the asset's fundamental value. An AI chatbot, conversely, bypasses this bottleneck by aggregating information directly from user queries as in Proposition \ref{Prop_AI_Learning}.
 
\paragraph{Information Aggregation with AI Chatbots and Experts.} Having established that the AI perfectly learns $\theta$, we now consider the scenario where the AI publishes this fundamental directly to traders. Unlike Section \ref{Additional Information y Model}, where the public report, e.g., a published private signal $x_j$ of a trader $j$,\footnote{\label{fn:ReportAsSignal}If we interpret the report $r=\theta+\sigma_{\varepsilon}\varepsilon$ as a published signal $x_j=\theta+\sigma_x\xi_j$ of a trader $j$, we would have $\sigma_x=\sigma_{\varepsilon}$ and $\xi_j=\varepsilon$.} contained common noise $\varepsilon$, the AI now, given that it learns $\theta$ from traders' queries, publishes $\theta$, which each trader reads with idiosyncratic noise, $y_i = \theta + \sigma_y \tau_i$. If, instead, the chatbot's answer carries a common error, e.g., algorithmic approximation errors (``hallucinations''), the AI economy is the economy of Section \ref{Additional Information y Model}, with the chatbot in the analyst's role.
 
\begin{prop}\label{Prop_Unique_Equilibrium}
If the AI publishes $\theta$, which traders read as $y_i = \theta + \sigma_y \tau_i$, such that the common noise term is eliminated ($\sigma_\varepsilon = 0$), the price no longer contains common noise: the linear rational expectations equilibrium has $g=0$ and $b=-\gamma\sigma_\eta/\alpha_{xy}$, where $\alpha_{xy}:=\sigma_x^{-2}+\sigma_y^{-2}$ denotes traders' combined private precision. Furthermore, this equilibrium is strictly more informative than the no-expert benchmark of Proposition \ref{Hellwig1980}.
\end{prop}
 
\begin{proof}[Proof of Proposition \ref{Prop_Unique_Equilibrium}]
When $\sigma_\varepsilon=0$, the $\varepsilon$-condition reads $g(1-Q)=0$. Since
$\Lambda=(b^2+g^2)(\sigma_x^2+\sigma_y^2)+\sigma_x^2\sigma_y^2$ and $Q=\sigma_x^2\sigma_y^2/\Lambda$,
$Q=1$ iff $b=g=0$, the trivial solution, which by Remark \ref{Remark1} is not an equilibrium. Hence every equilibrium has $g=0$. Then
$1-Q=b^2(\sigma_x^2+\sigma_y^2)/\Lambda$ and $\mathrm{Var}=\sigma_x^2\sigma_y^2b^2/\Lambda$, so the
$\eta$-condition $b(1-Q)=-\gamma\sigma_\eta \mathrm{Var}$ becomes
$b(\sigma_x^2+\sigma_y^2)=-\gamma\sigma_\eta\sigma_x^2\sigma_y^2$, i.e. $b=-\gamma\sigma_\eta/\alpha_{xy}$.
 
The price signal $Z$ is therefore entirely purged of interpretation noise, leaving only the aggregate supply shock $\eta$. Because both private precision and price precision are strictly higher, $\alpha_{xy} > \sigma_x^{-2}$ and $\alpha_{xy}^2/(\gamma\sigma_\eta)^2>1/(\gamma\sigma_x^2\sigma_\eta)^2$, the posterior variance satisfies:
\begin{equation}
    \mathrm{Var}(\theta \mid x_i, y_i, P) = \frac{1}{\alpha_{xy} + \left(\frac{\gamma \sigma_\eta}{\alpha_{xy}}\right)^{-2}} < \mathrm{Var}(\theta \mid x_i, P),
\end{equation}
where the right-hand side is the posterior variance \eqref{PosteriorVarianceBenchmark} of the benchmark economy. The publication of precise AI knowledge thus removes common noise from prices and induces a strictly more informative equilibrium.
\end{proof}
 
Taken together, Propositions \ref{Prop_AI_Learning} and \ref{Prop_Unique_Equilibrium} establish the central baseline results: (i) the chatbot learns $\theta$ efficiently directly from traders' queries, and (ii) once it communicates this knowledge without common noise ($\sigma_\varepsilon=0$), prices no longer load on common noise ($g=0$).
 
Proposition \ref{Prop_Informational_Dominance} below strengthens this result by showing that eliminating common noise ($\sigma_\varepsilon = 0$) does not merely improve upon the no-expert baseline of Proposition \ref{Hellwig1980}, but strictly dominates the equilibrium of the noisy expert economy ($\sigma_\varepsilon > 0$).
 
\begin{prop}\label{Prop_Informational_Dominance}
The equilibrium of the AI economy without common noise ($\sigma_\varepsilon = 0$) is strictly more informative than the equilibrium $(b,g)$ of the economy with common noise ($\sigma_\varepsilon > 0$). That is, at the unique root $g \in (0, \bar{g})$ of \eqref{CubicG} established in Proposition \ref{Prop_Uniqueness_Noisy}, traders' equilibrium posterior variance satisfies:
\begin{equation}
\mathrm{Var}(\theta \mid x_i, y_i, P)_{\sigma_\varepsilon = 0} < \mathrm{Var}(\theta \mid x_i, y_i, P)_{\sigma_\varepsilon > 0}. 
\end{equation}
\end{prop}
\begin{proof}
    We prove Proposition \ref{Prop_Informational_Dominance} in Appendix B.
\end{proof}
 
\begin{coro}\label{Cor_Expert}
Let $\sigma_\varepsilon=\sigma_x$, as when the report is the published signal $x_j$ of a trader $j$ in footnote \ref{fn:ReportAsSignal}, and let $(b,g)$ be the equilibrium of the
economy with common noise. The expert who publishes the report $r$ has posterior variance
\begin{equation}
\mathrm{Var}(\theta\mid r,P)_{\sigma_\varepsilon=\sigma_x}
=\frac{\sigma_x^2 b^2}{b^2+(\sigma_x-g)^2}, \label{Var_Expert}
\end{equation}
and
\begin{equation}
\mathrm{Var}(\theta\mid x_i,y_i,P)_{\sigma_\varepsilon=0}
<\mathrm{Var}(\theta\mid x_i,y_i,P)_{\sigma_\varepsilon=\sigma_x}
<\mathrm{Var}(\theta\mid r,P)_{\sigma_\varepsilon=\sigma_x}. \label{Chain_Expert}
\end{equation}
That is, traders who read the expert's report are better informed than the expert himself,
and every trader in the AI economy is better informed than both.
\end{coro}
 
\begin{proof}
Let $s:=\sigma_\varepsilon=\sigma_x$. The expert observes $r=\theta+s\varepsilon$ and, by
\eqref{Zsignal}, the price signal $Z=\theta+b\eta+g\varepsilon$. Since $0<g<s$
by Lemma \ref{GBoundLemma}, he can remove the common noise from the price by forming
\[
\frac{Z-\frac{g}{s}r}{1-\frac{g}{s}}=\theta+\frac{b}{1-g/s}\,\eta,
\]
which is independent of the noise in $r$. His posterior precision is therefore
\[
\Omega_E:=\frac{1}{\mathrm{Var}(\theta\mid r,P)}
=\frac{1}{s^2}+\frac{(s-g)^2}{s^2b^2}
=\frac{b^2+(s-g)^2}{s^2b^2},
\]
which yields \eqref{Var_Expert}.
 
The first inequality in \eqref{Chain_Expert} is Proposition \ref{Prop_Informational_Dominance}.
For the second, let $M:=b^2(s^2+\sigma_y^2)+g^2\sigma_y^2>0$. By \eqref{Omega_noisy_app},
\[
\Omega_{\text{noisy}}-\Omega_E
=\frac{1}{s^2}+\frac{b^2+(s-g)^2+\sigma_y^2}{M}-\frac{b^2+(s-g)^2}{s^2b^2}
=\frac{b^4s^2+\sigma_y^2\, g\left[b^2(2s-g)-g(s-g)^2\right]}{s^2b^2M}.
\]
At $\sigma_\varepsilon=\sigma_x$, the equilibrium relation \eqref{BfromG} reads
\[
b^2=\frac{g\left[s^2(s-g)^2+\sigma_y^2g^2\right]}{s^3-(2s^2+\sigma_y^2)g},
\]
where the denominator is positive by Lemma \ref{GBoundLemma}. Substituting this expression,
the inequality $b^2(2s-g)>g(s-g)^2$ is equivalent to
\[
s^2\left[(s-g)^2(s+g)+\sigma_y^2g\right]>0,
\]
which holds for every $g>0$. Hence the numerator above is strictly positive,
$\Omega_{\text{noisy}}>\Omega_E$, and taking reciprocals establishes the second inequality
in \eqref{Chain_Expert}.
\end{proof}
 
\noindent Appendix C shows that readers of a report are also better informed than its author when all traders read the report identically ($\sigma_y=0$): their posterior precision then exceeds that of the author by exactly the precision $1/\sigma_x^2$ of their private signal.

\section{Trading Volume}\label{Section Volume}
 
The empirical motivation for differential interpretation comes from trading volume: \citet{Kan95} and \citet{SHENG2015515} show that public announcements generate trade because different investors read them differently. We now compute the volume that each of our three economies generates: the \citet{Gro76,Gro80} benchmark of Section \ref{Section Model Perfect Information Aggregation}, the economy with an analyst report of Section \ref{Additional Information y Model}, and the AI economy of Section \ref{Section Information Centralization}.
 
We assume that traders enter the market without initial holdings, such that trader $i$ trades $k_i$ shares and noise traders trade $K^S=\sigma_\eta\eta$ shares. Since every share bought is also sold, expected trading volume is
\begin{eqnarray}
\mathcal{V}:=\frac{1}{2}E\left[\int_0^1|k_i|\,di+|\sigma_\eta\eta|\right]. \label{VolumeDef}
\end{eqnarray}
Volume depends on how strongly traders' demands disagree. We measure this disagreement by the cross-sectional dispersion of individual demand around aggregate demand,
\begin{eqnarray}
\Sigma_k:=\mathrm{Var}\left(k_i-\int_0^1k_j\,dj\right). \label{DispersionDef}
\end{eqnarray}
 
\begin{prop}\label{Prop_Volume}
In each economy, expected trading volume is
\begin{eqnarray}
\mathcal{V}=\frac{1}{\sqrt{2\pi}}\left(\sigma_\eta+\sqrt{\sigma_\eta^2+\Sigma_k}\right), \label{VolumeFormula}
\end{eqnarray}
where the dispersion of demand is
\begin{eqnarray*}
\Sigma_k^{GS}=\frac{1}{\gamma^2\sigma_x^2}, \qquad
\Sigma_k^{E}=\frac{1}{\gamma^2\sigma_x^2}\left[1+\frac{\sigma_y^2}{\sigma_x^2}\left(\frac{g}{\sigma_\varepsilon-g}\right)^2\right], \qquad
\Sigma_k^{AI}=\frac{1}{\gamma^2}\left(\frac{1}{\sigma_x^2}+\frac{1}{\sigma_y^2}\right)
\end{eqnarray*}
in the benchmark economy without experts, in the equilibrium of the economy with an analyst report ($\sigma_\varepsilon>0$), and in the equilibrium of the AI economy ($\sigma_\varepsilon=0$), respectively. Moreover,
\begin{eqnarray}
\Sigma_k^{GS}<\Sigma_k^{E}<\Sigma_k^{AI}
\quad\text{and hence}\quad
\mathcal{V}^{GS}<\mathcal{V}^{E}<\mathcal{V}^{AI}. \label{VolumeRanking}
\end{eqnarray}
\end{prop}
 
\begin{proof}
\textit{Demand weights.} Consider the economy with an analyst report and let $M:=b^2(\sigma_\varepsilon^2+\sigma_y^2)+g^2\sigma_y^2$. By \eqref{determinant}, $\Lambda=M+\sigma_x^2\left(b^2+g^2+\sigma_\varepsilon^2+\sigma_y^2-2\sigma_\varepsilon g\right)$, and by \eqref{Posterior Variance}, $\mathrm{Var}(\theta|x_i,y_i,Z)=\sigma_x^2M/\Lambda$. By \eqref{Posterior mean}, the posterior mean places weight
\[
1-\frac{\sigma_x^2}{\Lambda}\left(b^2+g^2+\sigma_\varepsilon^2+\sigma_y^2-2\sigma_\varepsilon g\right)=\frac{M}{\Lambda}
\]
on $x_i$ and weight $\frac{\sigma_x^2}{\Lambda}\left(b^2+g^2-\sigma_\varepsilon g\right)$ on $y_i$. Dividing by $\gamma \mathrm{Var}(\theta|x_i,y_i,Z)$, demand \eqref{demand} responds to $x_i$ and $y_i$ with coefficients
\begin{eqnarray}
\frac{\partial k_i}{\partial x_i}=\frac{1}{\gamma\sigma_x^2}, \qquad
\frac{\partial k_i}{\partial y_i}=\frac{\rho}{\gamma}, \qquad
\rho:=\frac{b^2+g^2-\sigma_\varepsilon g}{M}. \label{DemandWeights}
\end{eqnarray}
 
\textit{Dispersion and volume.} All remaining terms in $k_i$ are common to all traders. Since $x_i=\theta+\sigma_x\xi_i$ and $y_i=\theta+\sigma_\varepsilon\varepsilon+\sigma_y\tau_i$, the idiosyncratic part of demand is
\[
k_i-\int_0^1k_j\,dj=\frac{\xi_i}{\gamma\sigma_x}+\frac{\rho\,\sigma_y\tau_i}{\gamma},
\qquad\text{such that}\qquad
\Sigma_k=\frac{1}{\gamma^2}\left(\frac{1}{\sigma_x^2}+\rho^2\sigma_y^2\right).
\]
Market clearing implies $\int_0^1k_j\,dj=\sigma_\eta\eta$. As $\xi_i$ and $\tau_i$ are independent of $\eta$, each $k_i$ is normally distributed with mean zero and variance $\sigma_\eta^2+\Sigma_k$, such that $E|k_i|=\sqrt{2/\pi}\sqrt{\sigma_\eta^2+\Sigma_k}$. Together with $E|\sigma_\eta\eta|=\sigma_\eta\sqrt{2/\pi}$, definition \eqref{VolumeDef} yields \eqref{VolumeFormula}, which is strictly increasing in $\Sigma_k$.
 
\textit{Benchmark and AI economy.} Without experts, traders hold no signal $y_i$, and \eqref{demand1} yields $\partial k_i/\partial x_i=\alpha_x/\gamma=1/(\gamma\sigma_x^2)$, which yields $\Sigma_k^{GS}$. In the AI economy, $\sigma_\varepsilon=0$ and, by Proposition \ref{Prop_Unique_Equilibrium}, $g=0$, such that $M=b^2\sigma_y^2$ and $\rho=1/\sigma_y^2$, which yields $\Sigma_k^{AI}$.
 
\textit{Economy with an analyst report.} Substituting the equilibrium relation \eqref{BfromG} and writing $D(g)=(\sigma_\varepsilon^2+\sigma_x^2+\sigma_y^2)g-\sigma_\varepsilon\sigma_x^2$ as in Lemma \ref{GBoundLemma}, we obtain
\[
b^2+g^2-\sigma_\varepsilon g=\frac{\sigma_\varepsilon g^2\left(\sigma_\varepsilon^2-\sigma_\varepsilon g+\sigma_y^2\right)}{-D(g)},
\qquad
M=\frac{\sigma_\varepsilon g\,\sigma_x^2(\sigma_\varepsilon-g)\left(\sigma_\varepsilon^2-\sigma_\varepsilon g+\sigma_y^2\right)}{-D(g)},
\]
and hence
\begin{eqnarray}
\rho=\frac{g}{\sigma_x^2(\sigma_\varepsilon-g)}. \label{RhoEquilibrium}
\end{eqnarray}
Substituting \eqref{RhoEquilibrium} into $\Sigma_k$ yields $\Sigma_k^{E}$.
 
\textit{Ranking.} By Lemma \ref{GBoundLemma}, $0<g<\sigma_\varepsilon$, so $\rho>0$ and $\Sigma_k^{E}>\Sigma_k^{GS}$. Moreover, $\rho<1/\sigma_y^2$ if and only if $g(\sigma_x^2+\sigma_y^2)<\sigma_\varepsilon\sigma_x^2$. This holds because
\[
g<\bar g=\frac{\sigma_\varepsilon\sigma_x^2}{\sigma_\varepsilon^2+\sigma_x^2+\sigma_y^2}<\frac{\sigma_\varepsilon\sigma_x^2}{\sigma_x^2+\sigma_y^2}.
\]
Since $\rho^2\sigma_y^2<1/\sigma_y^2$, it follows that $\Sigma_k^{E}<\Sigma_k^{AI}$. As $\mathcal{V}$ is strictly increasing in $\Sigma_k$, \eqref{VolumeRanking} follows.
\end{proof}
 
\noindent Regarding Proposition \ref{Prop_Volume} we remark:
\begin{remark}\label{Remark_Volume}
Differential interpretation of an analyst report raises trading volume above the \citet{Gro76,Gro80} benchmark, in line with the evidence of \citet{Kan95} and \citet{SHENG2015515}. Traders discount their reading of the report, since its common error is partly reflected in the price: by \eqref{DemandWeights} and \eqref{RhoEquilibrium}, the ratio of their demand sensitivities to $y_i$ and $x_i$ is $g/(\sigma_\varepsilon-g)$, compared with the larger ratio $\sigma_x^2/\sigma_y^2$ in the AI economy. The AI's answer carries no common noise, such that traders act more aggressively on their reading $y_i$, and volume rises further. Prices are ranked like volume: $1/(\gamma\sigma_x^2\sigma_\eta)^2<1/(b^2+g^2)<\alpha_{xy}^2/(\gamma\sigma_\eta)^2$, where the outer terms are the precisions of the benchmark price and of the AI price, and both inequalities are shown in Appendix B. Hence, higher trading volume need not signal less informative markets.
\end{remark}

\section{Limitations, Detailed Literature Comparison, and Directions for Future Research}\label{Limitations and directions for future research}
The current highly stylized model relies on several strong assumptions. These assumptions help to connect the classic literature on information aggregation in large markets to the current literature on the use of AI-based tools in financial markets. At the same time, these assumptions limit the applicability of our model. In the following, we discuss some of these limitations.

\paragraph{Comparison to Earlier Models} As in \citet{Hel80}, Proposition \ref{Prop_Uniqueness_Noisy} shows that the linear equilibrium is unique, even though prices now contain common interpretation noise. In \citet{harris1993differences}, it is common knowledge that two groups of traders believe that a public report $r=\theta+\sigma_\varepsilon\varepsilon$ has noise $\sigma^a_\varepsilon$ and $\sigma^b_\varepsilon$, respectively, such that the groups weight the same value of $r$ differently. In our model, traders agree on $\sigma_\varepsilon$, but each trader reads the report with private noise, $y_i=r+\sigma_y\tau_i$, such that his reading is private information and the report's common error enters prices. Furthermore, Sections \ref{Remarks and Extensions}, \ref{Section Information Centralization}, and \ref{Section Volume} explore AI chatbots as a second channel of information centralization alongside traditional markets and prices. By modeling an AI platform that directly aggregates private queries, we address a mechanism of information aggregation that did not exist in this earlier literature.

\paragraph{Market Structure} We assume a single dominant chatbot, similar to Google in search: developing and training a chatbot involves large fixed costs, while the marginal cost of answering an additional query is close to zero. An oligopoly of chatbots that compete for trader queries would require a model of platform competition, for which \citet{Roc03}, \citet{Bar03b}, \citet{Ori15}, \citet{Bak19}, and \citet{Bud24} may provide starting points.
 
\paragraph{Rationality and Signal Structures} The present model studies the role of experts and chatbots within a model in which agents are rational Bayesians. That is, the present model aligns with the literature that has developed from the \cite{Gro76,Gro80} model on information aggregation. Our results hold for all $\sigma_x,\sigma_y>0$. They therefore also cover the case in which the precision with which traders interpret a report is linked to the precision of their private information, $\sigma_y=h(\sigma_x)$ for some function $h$. Richer structures, in which traders differ in $\sigma_x$ and $\sigma_y$, or in which interpretation errors $\tau_i$ are correlated with private signal errors $\xi_i$, change traders' posterior weights. Since these errors are idiosyncratic, however, they continue to average out in the aggregate, and the report's common error $\varepsilon$ remains the only aggregate component of interpretation noise. We leave a full analysis of such structures to future research. 
 
In future research, one may study the role of experts and chatbots in models where traders are not fully rational. In particular, the present approach can be combined with the behavioral models of \citet{Shi10}, \citet{Shi17}, \citet{kuhle2022thought}, and \citet{Sun24} which emphasize herding and behavioral patterns in financial markets.
 
A further limitation of our model concerns its distributional assumptions. That is, to improve tractability, we follow \cite{Gro76,Gro80}, \citet{Gre75}, \citet{Hel80}, \cite{Ang06}, and \cite{Viv08}, and focus on normally distributed signal structures.\footnote{Our model thus aligns with ongoing empirical work of \cite{Kac24}, \citet{Dav25}, and \citet{McC25}, who use the model of \cite{Gro76,Gro80} as a starting point; see also the related literature on information disclosure, e.g. \citet{WAN21}.} Thus, our model does not take into account empirically observable phenomena such as fat tails or the skewness of returns. 
 
\paragraph{Further Aspects} The present model is static and assumes that chatbots learn instantaneously from traders' queries. A natural extension would be to consider dynamic settings in which chatbots are pre-trained on historical query data and continuously update their knowledge over time. The techniques developed in \cite{Viv93} and \cite{Viv08} could provide a useful starting point for such dynamic models. Moreover, we take traders' submission of queries as given. Endogenizing traders' decision whether to consult the chatbot, and the resulting query-participation equilibrium, is beyond the scope of the present paper. Corollary \ref{Cor_Expert} shows that the readers of a report can be better informed than the expert who wrote it. A full analysis of experts' incentives to publish is also left for future research.
 
Furthermore, the model highlights that chatbots (i) centralize traders' dispersed private information and (ii) thereby gain an informational advantage over market participants. While the provision of free advice appears beneficial, the centralization of information can create inefficiencies. In particular, a dominant chatbot may render markets more susceptible to fraud and manipulation if agents rely heavily on its recommendations. For example, even if the chatbot provides correct information most of the time, occasional misleading signals could generate arbitrage opportunities for its owners. Similarly, a chatbot that wants to preserve its informational advantage may have an incentive to add common noise to its answers.
 
\section{Conclusion}\label{Conclusion}

The present model of information aggregation incorporates the empirical observation of \citet{Kan95}, \cite{Kandel99}, \cite{SHENG2015515} and others that publicly available reports are subject to differential interpretation. Unlike previous theoretical models, the present framework does not rely on public disagreement or on heterogeneous priors to model the differential interpretation of a given piece of information. Instead, we assume that each agent who reads a given news item forms his own \textit{private} opinion as to its meaning. Once analyst reports are subject to differential interpretation, both noise-trader activity and interpretation noise hinder information aggregation through prices.

Building on our baseline model, we consider the role of ``AI’’ chatbots, which provide free financial advice and analyze traders’ queries to learn about the asset’s fundamental. Our model predicts that chatbots are not constrained by noise-trader activity and thus aggregate information more efficiently than markets. Moreover, when chatbots publish what they have learned, their answers enhance the informational content of prices. Ultimately, the AI platform gains an absolute informational advantage over market participants, suggesting a shift from information aggregation through market prices to aggregation by AI platforms.

Finally, differential interpretation of an analyst report raises trading volume above the \citet{Gro76,Gro80} benchmark, in line with the evidence of \citet{Kan95} and \citet{SHENG2015515}. The AI economy generates the highest trading volume while also featuring the most informative prices, such that high trading volume need not signal poorly informed markets.
 
\newpage
\addcontentsline{toc}{section}{References}
\markboth{References}{References}
\bibliographystyle{apalike}
\bibliography{References.bib}
 
\newpage
 
 \section*{Appendix A: Derivation of the Posteriors and of Aggregate Demand}
 
This appendix derives the posteriors \eqref{Posterior mean}--\eqref{determinant}, aggregate demand \eqref{AggregateDemand}, and the two equilibrium conditions used in the proof of Proposition \ref{Two equilibria Proposition}.
 
\paragraph{Information structure.} Given the price guess \eqref{guess2}, trader $i$ observes the three signals $\mathbf{s}_i:=(x_i,y_i,Z)'$, which we write as $\mathbf{s}_i=\theta\mathbf{1}+\mathbf{e}_i$, where $\mathbf{1}:=(1,1,1)'$ and
\begin{equation*}
\mathbf{e}_i:=\left(\sigma_x\xi_i,\ \sigma_\varepsilon\varepsilon+\sigma_y\tau_i,\ \hat b\eta+\hat g\varepsilon\right)', \qquad \hat b:=\frac{b}{a},\quad \hat g:=\frac{g}{a}.
\end{equation*}
Traders do not observe $r=\theta+\sigma_\varepsilon\varepsilon$ itself, such that $\varepsilon$ enters only through $y_i$ and $Z$. Since $\xi_i$, $\tau_i$, $\varepsilon$, and $\eta$ are independent standard normal random variables, which are also independent of $\theta$, the vector $\mathbf{e}_i$ is normally distributed with mean zero, independently of $\theta$, and with covariance matrix
\begin{equation*}
\Sigma=\begin{pmatrix}\sigma_x^2&0&0\\0&V_y&\kappa \\0&\kappa&V_z\end{pmatrix},\qquad V_y:=\sigma_\varepsilon^2+\sigma_y^2,\quad V_z:=\hat b^2+\hat g^2,\quad \kappa:=\sigma_\varepsilon\hat g.
\end{equation*}
For $a=1$, $V_y$, $V_z$, and $\kappa$ coincide with the definitions in the proof of Proposition \ref{Two equilibria Proposition}. The determinant of the lower-right block of $\Sigma$ is
\begin{equation*}
M:=V_yV_z-\kappa^2=\hat b^2\left(\sigma_\varepsilon^2+\sigma_y^2\right)+\hat g^2\sigma_y^2,
\end{equation*}
which is strictly positive whenever $(b,g)\neq(0,0)$, and which coincides for $a=1$ with $M$ in the proof of Proposition \ref{Prop_Volume}.
 
\paragraph{Posteriors.} Under the uniform prior over $\theta$, the posterior of $\theta$ given $\mathbf{s}_i$ is normal with precision $\mathbf{1}'\Sigma^{-1}\mathbf{1}$ and mean $\mathbf{1}'\Sigma^{-1}\mathbf{s}_i/\mathbf{1}'\Sigma^{-1}\mathbf{1}$. Since $\Sigma$ is block diagonal,
\begin{equation*}
\Sigma^{-1}=\begin{pmatrix}1/\sigma_x^2&0&0\\0&V_z/M&-\kappa/M\\0&-\kappa/M&V_y/M\end{pmatrix},
\qquad
\mathbf{1}'\Sigma^{-1}\mathbf{1}=\frac{1}{\sigma_x^2}+\frac{V_y+V_z-2\kappa }{M}=\frac{\Lambda}{\sigma_x^2M},
\end{equation*}
where $\Lambda:=M+\sigma_x^2\left(V_y+V_z-2\kappa\right)$ is exactly \eqref{determinant}. Because $V_y+V_z-2\kappa =\hat b^2+(\sigma_\varepsilon-\hat g)^2+\sigma_y^2>0$, we have $\Lambda>0$. The posterior variance is therefore
\begin{equation*}
\mathrm{Var}(\theta|x_i,y_i,Z)=\frac{\sigma_x^2M}{\Lambda}=\frac{\sigma_x^2}{\Lambda}\left[V_yV_z-\kappa^2\right],
\end{equation*}
which is \eqref{Posterior Variance}. The posterior mean places the weights $w:=\Sigma^{-1}\mathbf{1}/\mathbf{1}'\Sigma^{-1}\mathbf{1}$ on $(x_i,y_i,Z)$, i.e.,
\begin{equation*}
w_x=\frac{M}{\Lambda},\qquad w_y=\frac{\sigma_x^2\left(V_z-\kappa\right)}{\Lambda},\qquad w_Z=\frac{\sigma_x^2\left(V_y-\kappa\right)}{\Lambda},\qquad w_x+w_y+w_Z=1.
\end{equation*}
Hence,
\begin{equation*}
E[\theta|x_i,y_i,Z]=w_xx_i+w_yy_i+w_ZZ=x_i+w_y\left(y_i-x_i\right)+w_Z\left(Z-x_i\right),
\end{equation*}
which is \eqref{Posterior mean}.
 
\paragraph{Aggregation.} The posterior variance is identical across traders, such that aggregate demand is $K^D=\left(\int_0^1E[\theta|x_i,y_i,Z]\,di-P\right)/\left(\gamma \mathrm{Var}(\theta|x_i,y_i,Z)\right)$. By the exact law of large numbers, $\int_0^1x_i\,di=\theta$ and $\int_0^1y_i\,di=\theta+\sigma_\varepsilon\varepsilon$, whereas $Z$ is common to all traders. Therefore, both adjustment terms in \eqref{Posterior mean} contribute to aggregate beliefs:
\begin{equation*}
\int_0^1\left(y_i-x_i\right)di=\sigma_\varepsilon\varepsilon,\qquad \int_0^1\left(Z-x_i\right)di=Z-\theta=\hat b\eta+\hat g\varepsilon,
\end{equation*}
such that
\begin{equation*}
\int_0^1E[\theta|x_i,y_i,Z]\,di=\theta+w_y\sigma_\varepsilon\varepsilon+w_Z\left(\hat b\eta+\hat g\varepsilon\right),
\end{equation*}
which is the numerator of \eqref{AggregateDemand}.
 
\paragraph{Market clearing and coefficient matching.} Market clearing, $K^D=\sigma_\eta\eta$, yields
\begin{equation*}
P=\theta+w_y\sigma_\varepsilon\varepsilon+w_Z\left(\hat b\eta+\hat g\varepsilon\right)-\gamma\sigma_\eta \mathrm{Var}(\theta|x_i,y_i,Z)\,\eta,
\end{equation*}
which is \eqref{Price2}. Comparing this price with the guess \eqref{guess2}, the coefficient on $\theta$ yields $a=1$, such that $\hat b=b$ and $\hat g=g$, and the constant yields $c=0$. With $Q:=w_Z$, the coefficients on $\eta$ and $\varepsilon$ yield
\begin{equation*}
b\left(1-Q\right)=-\gamma\sigma_\eta \mathrm{Var}(\theta|x_i,y_i,Z),\qquad g\left(1-Q\right)=w_y\sigma_\varepsilon=\frac{\sigma_x^2\sigma_\varepsilon\left(V_z-\kappa\right)}{\Lambda},
\end{equation*}
which are the two equilibrium conditions in the proof of Proposition \ref{Two equilibria Proposition}. Propositions \ref{Two equilibria Proposition} and \ref{Prop_Uniqueness_Noisy} show that, in addition to $b=g=0$, this system has exactly one solution with $(b,g)\neq(0,0)$.
 
\newpage
 
 \section*{Appendix B: Proof of Proposition \ref{Prop_Informational_Dominance}}
 
Let $\alpha_{xy} := \sigma_x^{-2} + \sigma_y^{-2}$ denote aggregate private precision when $\sigma_\varepsilon = 0$. By Proposition \ref{Prop_Unique_Equilibrium}, in the AI equilibrium without common noise ($\sigma_\varepsilon = 0$), $g = 0$ and the price noise loading is $b_{AI} := -\gamma \sigma_\eta / \alpha_{xy}$. The total posterior precision in this equilibrium is:
\begin{equation}
\Omega_{AI} := \frac{1}{\mathrm{Var}(\theta \mid x_i, y_i, P)_{\sigma_\varepsilon = 0}} = \alpha_{xy} + \frac{1}{b_{AI}^2} = \alpha_{xy} + \frac{\alpha_{xy}^2}{\gamma^2 \sigma_\eta^2}. \label{Omega_AI_app}
\end{equation}
 
Now consider the equilibrium of the economy with common noise ($\sigma_\varepsilon > 0$), characterized by price coefficients $(b, g)$ where $g \in (0, \bar{g})$ satisfies \eqref{CubicG}. From \eqref{Posterior mean}--\eqref{determinant}, traders' posterior precision $\Omega_{\text{noisy}} := 1 / \mathrm{Var}(\theta \mid x_i, y_i, Z)$ is given by:
\begin{equation}
\Omega_{\text{noisy}} = \frac{1}{\sigma_x^2} + \frac{b^2 + (g - \sigma_\varepsilon)^2 + \sigma_y^2}{b^2 (\sigma_\varepsilon^2 + \sigma_y^2) + g^2 \sigma_y^2}. \label{Omega_noisy_app}
\end{equation}
 
By direct algebraic decomposition, the second term in \eqref{Omega_noisy_app} satisfies:
\begin{equation}
\frac{b^2 + (g - \sigma_\varepsilon)^2 + \sigma_y^2}{b^2 (\sigma_\varepsilon^2 + \sigma_y^2) + g^2 \sigma_y^2} = \frac{1}{\sigma_y^2} + \frac{1}{b^2} - \Delta(b, g), \label{Decomp_app}
\end{equation}
where
\begin{equation}
\Delta(b, g) := \frac{(b^2 \sigma_\varepsilon + g \sigma_y^2)^2}{b^2 \sigma_y^2 \left[ b^2 (\sigma_\varepsilon^2 + \sigma_y^2) + g^2 \sigma_y^2 \right]}.
\end{equation}
Since $\sigma_\varepsilon > 0$, $\sigma_y > 0$, and, by Lemma \ref{GBoundLemma}, $g > 0$ and $b^2 > 0$ in equilibrium, $\Delta(b, g) > 0$. Thus, \eqref{Omega_noisy_app} can be written as:
\begin{equation}
\Omega_{\text{noisy}} = \alpha_{xy} + \frac{1}{b^2} - \Delta(b, g). \label{Omega_noisy_decomp}
\end{equation}
 
Next, we bound $b^2$. Matching the supply shock coefficient in market clearing \eqref{Price2} yields $b(1 - Q) = -\gamma \sigma_\eta \mathrm{Var}(\theta \mid x_i, y_i, Z) = -\gamma \sigma_\eta / \Omega_{\text{noisy}}$, where $Q := \frac{\sigma_x^2}{\Lambda}(\sigma_\varepsilon^2 + \sigma_y^2 - \sigma_\varepsilon g)$. Multiplying $(1 - Q)$ by $\Omega_{\text{noisy}} = \frac{\Lambda}{\sigma_x^2 [b^2(\sigma_\varepsilon^2 + \sigma_y^2) + g^2 \sigma_y^2]}$ yields:
\begin{equation}
(1 - Q)\Omega_{\text{noisy}} = \frac{1}{\sigma_x^2} + \frac{b^2 + g^2 - \sigma_\varepsilon g}{b^2(\sigma_\varepsilon^2 + \sigma_y^2) + g^2 \sigma_y^2}.
\end{equation}
Comparing this directly to $\alpha_{xy} = \sigma_x^{-2} + \sigma_y^{-2}$, we find:
\begin{equation}
(1 - Q)\Omega_{\text{noisy}} - \alpha_{xy} = \frac{b^2 + g^2 - \sigma_\varepsilon g}{b^2(\sigma_\varepsilon^2 + \sigma_y^2) + g^2 \sigma_y^2} - \frac{1}{\sigma_y^2} = -\frac{\sigma_\varepsilon (g \sigma_y^2 + b^2 \sigma_\varepsilon)}{\sigma_y^2 \left[ b^2 (\sigma_\varepsilon^2 + \sigma_y^2) + g^2 \sigma_y^2 \right]} < 0.
\end{equation}
Hence, $(1 - Q)\Omega_{\text{noisy}} < \alpha_{xy}$. 
 
Before squaring, we verify that $(1-Q)\Omega_{\text{noisy}}\geq 0$, so that both sides of the inequality above are nonnegative. Substituting the equilibrium relation \eqref{BfromG} into the definition of $Q$ and simplifying yields, for every $g\in(0,\bar g)$,
\begin{equation}
1-Q = \frac{-g\,\sigma_\varepsilon^2}{g(\sigma_x^2+\sigma_y^2)-\sigma_\varepsilon\sigma_x^2}. \label{OneMinusQ_app}
\end{equation}
By Lemma \ref{GBoundLemma}, $D(g):=(\sigma_\varepsilon^2+\sigma_x^2+\sigma_y^2)g-\sigma_\varepsilon\sigma_x^2<0$ for every $g\in(0,\bar g)$. Since the denominator in \eqref{OneMinusQ_app} equals $D(g)-\sigma_\varepsilon^2 g<D(g)<0$, and its numerator $-g\sigma_\varepsilon^2$ is strictly negative, we obtain $1-Q>0$. Because $b(1-Q)=-\gamma\sigma_\eta\,\mathrm{Var}(\theta\mid x_i,y_i,Z)<0$ (the right-hand side being a product of strictly positive terms) and $1-Q>0$, it follows that $b<0$. Hence $(1-Q)\Omega_{\text{noisy}}=-\gamma\sigma_\eta/b>0$, and squaring the inequality $(1-Q)\Omega_{\text{noisy}}<\alpha_{xy}$ below is valid, as both sides are positive.
 
Squaring both sides and using $b^2 = \frac{\gamma^2 \sigma_\eta^2}{(1 - Q)^2 \Omega_{\text{noisy}}^2}$, we obtain:
\begin{equation}
\frac{1}{b^2} = \frac{(1 - Q)^2 \Omega_{\text{noisy}}^2}{\gamma^2 \sigma_\eta^2} < \frac{\alpha_{xy}^2}{\gamma^2 \sigma_\eta^2}. \label{b_bound_app}
\end{equation}
 
Combining \eqref{Omega_noisy_decomp} and \eqref{b_bound_app} with $\Delta(b, g) > 0$ yields:
\begin{equation}
\Omega_{\text{noisy}} = \alpha_{xy} + \frac{1}{b^2} - \Delta(b, g) < \alpha_{xy} + \frac{\alpha_{xy}^2}{\gamma^2 \sigma_\eta^2} = \Omega_{AI}.
\end{equation}
Taking reciprocals establishes $\mathrm{Var}(\theta \mid x_i, y_i, P)_{\sigma_\varepsilon > 0} = 1/\Omega_{\text{noisy}} > 1/\Omega_{AI} = \mathrm{Var}(\theta \mid x_i, y_i, P)_{\sigma_\varepsilon = 0}$ for the equilibrium $g \in (0, \bar{g})$.

\textit{Price precision in Remark \ref{Remark_Volume}.} The $\eta$-condition $b(1-Q)=-\gamma\sigma_\eta\mathrm{Var}(\theta|x_i,y_i,Z)$, with $\mathrm{Var}(\theta|x_i,y_i,Z)=\sigma_x^2M/\Lambda$ and \eqref{OneMinusQ_app}, expresses $\gamma^2\sigma_\eta^2$ in terms of $b$ and $g$. Substituting this expression and \eqref{BfromG} yields
\begin{equation*}
(\gamma\sigma_x^2\sigma_\eta)^2-(b^2+g^2)=\frac{g^2\sigma_\varepsilon\left[(\sigma_\varepsilon g+\sigma_x^2)(\sigma_\varepsilon-g)^2+g\,\sigma_\varepsilon\sigma_y^2\right]}{(\sigma_\varepsilon-g)^2\,\left(-D(g)\right)}>0,
\end{equation*}
since $0<g<\sigma_\varepsilon$ and $D(g)<0$ by Lemma \ref{GBoundLemma}. Hence $1/(\gamma\sigma_x^2\sigma_\eta)^2<1/(b^2+g^2)$. \hfill $\blacksquare$
 
\newpage
 
 \section*{Appendix C: Identical Interpretation of the Report ($\sigma_y=0$)}
 
This appendix contrasts the economy of Section \ref{Additional Information y Model}, in which traders interpret the report differently ($\sigma_y>0$), with an economy in which all traders read the report identically ($\sigma_y=0$). In the latter case, $y_i=r=\theta+\sigma_\varepsilon\varepsilon$ for all $i$, such that the report is exogenous public information that every trader observes.
 
\paragraph{Equilibrium.} We use the price guess \eqref{guess2}, $P=a\theta+b\eta+g\varepsilon+c$, and write $\hat b:=b/a$ and $\hat g:=g/a$ as in Appendix A. Since traders observe $r$, and since $\hat g\neq\sigma_\varepsilon$ in any equilibrium (see below), the price signal \eqref{Zsignal} is informationally equivalent, given $r$, to
\begin{equation*}
S:=\frac{Z-\frac{\hat g}{\sigma_\varepsilon}r}{1-\frac{\hat g}{\sigma_\varepsilon}}=\theta+\tilde b\eta,\qquad \tilde b:=\frac{\hat b}{1-\hat g/\sigma_\varepsilon}.
\end{equation*}
The errors of $x_i$, $r$, and $S$ are independent, such that trader $i$'s posterior precision and posterior mean are
\begin{equation*}
\Omega:=\frac{1}{\mathrm{Var}(\theta|x_i,r,P)}=\frac{1}{\sigma_x^2}+\frac{1}{\sigma_\varepsilon^2}+\frac{1}{\tilde b^2},\qquad
E[\theta|x_i,r,P]=\frac{1}{\Omega}\left(\frac{x_i}{\sigma_x^2}+\frac{r}{\sigma_\varepsilon^2}+\frac{S}{\tilde b^2}\right).
\end{equation*}
Integrating over traders, using $\int_0^1x_i\,di=\theta$, and imposing market clearing,
\begin{equation*}
\Omega\left(\int_0^1E[\theta|x_i,r,P]\,di-P\right)/\gamma=\sigma_\eta\eta,
\end{equation*}
yields
\begin{equation*}
P=\theta+\frac{1}{\sigma_\varepsilon\Omega}\,\varepsilon+\frac{1}{\Omega}\left(\frac{1}{\tilde b}-\gamma\sigma_\eta\right)\eta.
\end{equation*}
Matching coefficients with \eqref{guess2} gives $a=1$, $c=0$, $g=1/(\sigma_\varepsilon\Omega)$, and $b\,\Omega=1/\tilde b-\gamma\sigma_\eta$. Since $1-g/\sigma_\varepsilon=(\Omega-1/\sigma_\varepsilon^2)/\Omega$, we have $b\,\Omega=\tilde b\left(\Omega-1/\sigma_\varepsilon^2\right)=\tilde b\left(1/\sigma_x^2+1/\tilde b^2\right)$, such that the last condition becomes $\tilde b/\sigma_x^2=-\gamma\sigma_\eta$. Finally, $\hat g=\sigma_\varepsilon$ cannot occur in equilibrium: in that case, $Z-r=\hat b\eta$ carries no information about $\theta$, traders' posterior precision is $1/\sigma_x^2+1/\sigma_\varepsilon^2$, and the same steps yield a loading $g=\sigma_\varepsilon/(1+\sigma_\varepsilon^2/\sigma_x^2)<\sigma_\varepsilon$ on $\varepsilon$, a contradiction. Hence, the unique equilibrium with $\tilde b\neq0$ has
\begin{equation*}
\tilde b=-\gamma\sigma_x^2\sigma_\eta,\qquad g=\frac{1}{\sigma_\varepsilon\Omega}\in(0,\sigma_\varepsilon),\qquad b=\tilde b\left(1-\frac{g}{\sigma_\varepsilon}\right),\qquad \Omega=\frac{1}{\sigma_x^2}+\frac{1}{\sigma_\varepsilon^2}+\frac{1}{\left(\gamma\sigma_x^2\sigma_\eta\right)^2}.
\end{equation*}
 
\paragraph{Contrast with differential interpretation.} If the report is interpreted identically, prices still contain the report's error ($g>0$). However, since every trader observes $r$, traders filter the report out of the price. The price net of the report, $S=\theta-\gamma\sigma_x^2\sigma_\eta\eta$, is then exactly the benchmark price signal $\tilde Z$ of Proposition \ref{Hellwig1980}, and it has the same precision, $1/(\gamma\sigma_x^2\sigma_\eta)^2$. Correspondingly, traders' posterior precision $\Omega$ is the benchmark precision \eqref{PosteriorVarianceBenchmark} plus the precision $1/\sigma_\varepsilon^2$ of the report. That is, identically interpreted exogenous public information adds its own precision to traders' beliefs, but it does not change the precision of the price, which is endogenous public information, as a signal of $\theta$. With differential interpretation ($\sigma_y>0$), by contrast, traders do not observe $r$. The report's common error then remains an unobservable component of the price, whose noise consists of both $b\eta$ and $g\varepsilon$, with $g\in(0,\bar g)$ by Lemma \ref{GBoundLemma}.
 
\paragraph{The analyst and the readers of the report.} Corollary \ref{Cor_Expert} compares the analyst who writes the report, and who observes $r$ and $P$, with traders who read it. With $\sigma_y=0$, the analyst's information $(r,P)$ is equivalent to $(r,S)$, such that his posterior precision is $1/\sigma_\varepsilon^2+1/(\gamma\sigma_x^2\sigma_\eta)^2$. Traders additionally observe their private signal $x_i$, and their posterior precision $\Omega$ exceeds that of the analyst by exactly $1/\sigma_x^2$, for every $\sigma_\varepsilon>0$. This includes the case of footnote \ref{fn:ReportAsSignal}, in which the report is the published signal of a trader $j$ and $\sigma_\varepsilon=\sigma_x$. Hence, readers of a report are better informed than its author both when they interpret the report identically, for every $\sigma_\varepsilon>0$, and, by Corollary \ref{Cor_Expert}, when they interpret it differently and $\sigma_\varepsilon=\sigma_x$.

\hspace{2cm}
\newpage
\noindent\textbf{Abbreviations:}\\
\noindent AI: Artificial intelligence
\linespread{1}
\begin{table}[ht]
\centering
\footnotesize
\begin{tabular}{lp{12.5cm}}
\toprule
\textbf{Variable} & \textbf{Meaning} \\
\midrule
$\theta$ & Asset's unobservable fundamental return; uniform, uninformative prior \\
$i\in[0,1]$ & Measure-one continuum of agents (traders) \\
$\gamma$ & Coefficient of absolute risk aversion (CARA utility) \\
$P$ & Market-clearing (equilibrium) price of the asset \\
$K^S$ & Aggregate (stochastic) supply of shares \\
$K^D$ & Aggregate demand, $\int_0^1 k_i\,di$ \\
$k_i$ & Agent $i$'s individual asset demand \\
$\sigma_\eta\eta$ & Noise-trader (aggregate supply) shock; $\eta\sim\mathcal{N}(0,1)$ \\
\multicolumn{2}{l}{\textit{Baseline model (no experts)}} \\
$x_i$ & Agent $i$'s private signal, $x_i=\theta+\sigma_x\xi_i$ \\
$\sigma_x,\ \xi_i$ & Private-signal noise scale and idiosyncratic shock, $\xi_i\sim\mathcal{N}(0,1)$ \\
$a,b,c$ & Conjectured linear price coefficients, $P=a\theta+b\eta+c$ \\
$\tilde{Z}$ & Price-implied signal, $\tilde{Z}=(P-c)/a=\theta+\frac{b}{a}\eta$ \\
$\alpha_x,\ \alpha_{\tilde{Z}},\ \alpha$ & Precisions of $x_i$, of $\tilde{Z}$, and their sum $\alpha=\alpha_x+\alpha_{\tilde{Z}}$ (Section 2) \\
\multicolumn{2}{l}{\textit{Differential interpretation / expert-report model}} \\
$r$ & Analyst/expert report, $r=\theta+\sigma_{\varepsilon}\varepsilon$ \\
$\sigma_{\varepsilon},\ \varepsilon$ & Common report noise scale and shock (shared by all readers), $\varepsilon\sim\mathcal{N}(0,1)$ \\
$y_i$ & Agent $i$'s private reading of the report, $y_i=r+\sigma_y\tau_i$ \\
$\sigma_y,\ \tau_i$ & Idiosyncratic interpretation-noise scale and shock, $\tau_i\sim\mathcal{N}(0,1)$ \\
$g$ & Price loading on the common report noise $\varepsilon$ \\
$Z$ & Price-implied signal, $Z=(P-c)/a=\theta+\frac{b}{a}\eta+\frac{g}{a}\varepsilon$ \\
$\Lambda$ & Normalizing term in the Section 3 posteriors, $\Lambda=\det\Sigma/\mathrm{Var}(\theta|x_i,y_i,Z)$, where $\Sigma$ is the covariance matrix of the noise in $(x_i,y_i,Z)$ \\
$V_y,\ V_z,\ \kappa,\ Q$ & $V_y=\sigma_\varepsilon^2+\sigma_y^2$, $V_z=b^2+g^2$, $\kappa=\sigma_\varepsilon g$; $Q$ is the posterior weight on $Z$ \\
$A_3,A_2,A_1,A_0$ & Coefficients of the cubic equilibrium equation in $g$ \\
$\bar{g}$ & Upper bound on equilibrium $g$; $g\in(0,\bar{g})\subset(0,\sigma_{\varepsilon})$ \\
$\psi$ & Auxiliary parameter in the uniqueness proof, $\psi=\gamma^2\sigma_\eta^2\sigma_x^2(\sigma_\varepsilon^2+\sigma_y^2)/\sigma_\varepsilon^2$ \\
\multicolumn{2}{l}{\textit{AI chatbot model}} \\
$q_i=f(x_i)$ & Trader $i$'s query to the chatbot, an invertible function of $x_i$ \\
$f(\cdot),\,f^{-1}(\cdot)$ & Query-formulation map and its inverse (signal recovery by the AI) \\
$\mathit{AI}$ & Chatbot's answer, $\mathit{AI}=\theta$ (no common noise) \\
$y_i$ (AI) & Trader $i$'s private reading of the chatbot's answer, $y_i=\theta+\sigma_y\tau_i$ \\
$\tilde{q}_i,\ \sigma_q$ & Noisy signal the chatbot extracts from query $q_i$ (footnote variant), $\tilde{q}_i=x_i+\sigma_q\zeta_i$ \\
$\lambda$ & Share of traders submitting question type $A$ (coarse-query case) \\
$x^*$ & Signal cutoff separating question types $A$ and $B$ (coarse-query case) \\
$\Phi(\cdot)$ & Standard normal CDF (used to invert $\lambda$ for $\theta$) \\
$\mathcal{V}$ & Expected trading volume, $\mathcal{V}=\frac{1}{2}E[\int_0^1|k_i|di+|\sigma_\eta\eta|]$ \\
$\Sigma_k$ & Cross-sectional dispersion of individual demand around aggregate demand \\
$\rho$ & Demand sensitivity to $y_i$ (times $\gamma$); $\rho=g/(\sigma_x^2(\sigma_\varepsilon-g))$ in equilibrium \\
$\alpha_{xy}$ & Combined private precision once the AI publishes $y_i$ with $\sigma_{\varepsilon}=0$, $\alpha_{xy}=\sigma_x^{-2}+\sigma_y^{-2}$ \\
 
\bottomrule
\end{tabular}
\end{table}

\newpage

\end{document}